\documentclass[reqno]{amsart}

\usepackage{amsfonts}
\usepackage{amsmath}
\usepackage{amssymb}

\usepackage{subfig}
\usepackage{pdfsync}
\usepackage{graphicx}
\usepackage{epstopdf}
\usepackage{url}

\newtheorem{theorem}{Theorem}
\theoremstyle{plain}
\newtheorem{lemma}{Lemma}

\numberwithin{equation}{section}

\AtBeginDocument{{\noindent\small
This is a preprint of a paper whose final and definite form is with 
\emph{Commun. Fac. Sci. Univ. Ank. Ser. A1 Math. Stat.}, ISSN: 1303-5991. 
Submitted 08 Sept 2016; Article revised 15 Apr 2017;
Article accepted for publication 19 Apr 2017.} 
\vspace{9mm}}

\begin{document}
	
\title[Global Stability for a HIV/AIDS Model]{Global Stability for a HIV/AIDS Model}

\author{Cristiana J. Silva}
\address{Cristiana J. Silva: 
Center for Research and Development in Mathematics and Applications (CIDMA),
Department of Mathematics, University of Aveiro, 3810-193 Aveiro, Portugal.}
\email{cjoaosilva@ua.pt}

\author{Delfim F. M. Torres}
\address{Delfim F. M. Torres: 
Center for Research and Development in Mathematics and Applications (CIDMA),
Department of Mathematics, University of Aveiro, 3810-193 Aveiro, Portugal.}
\email{delfim@ua.pt}


\date{Received: September 08, 2016, Revised: April 15, 2017; Accepted: April 19, 2017.}



	
\begin{abstract}
We investigate global stability properties 
of a HIV/AIDS population model with constant recruitment rate, 
mass action incidence, and variable population size. Existence 
and uniqueness results for disease-free and endemic equilibrium 
points are proved. Global stability of the equilibria 
is obtained through Lyapunov's direct method 
and LaSalle's invariance principle. 
\end{abstract}
	

\maketitle

	
\section{Introduction}

Mathematical models may represent a useful tool in the development 
of public health policies \cite{Denysiuk:2015,MR3349757,MR3148910}. 
Although it is unlikely that a mathematical model will provide 
accurate long-term predictions on the number of AIDS cases, 
one such model, based on interactions that lead to disease transmission, 
could eventually allow researchers to answer many useful questions 
\cite{Hyman:Stanley:MathBio:1988}. As a result, several mathematical 
models have been proposed in the last decades for HIV/AIDS 
transmission dynamics: see, e.g.,
\cite{Abbas:etall:PlosOne2:2007,Anderson:1988,Anderson:etall:1986,Blower:etall:PhilTrans:1991,%
Hethcote:VanArk:1992,Mukandavire:Garirar:2007,Nyabadza:etall:NARWA:2011,Sani:etall:MathBiosc:2007} 
and references cited therein. 

Global stability of equilibrium points for mathematical models of HIV/AIDS 
transmission dynamics has been studied by different authors: see, e.g., \cite{Cai:etall:JCAM:2009,Feng:Rui:Yang:AMM:2016,Huo:Feng:AMM:2013}. 
In \cite{Huo:Feng:AMM:2013}, the authors consider different latent 
stages depending on other chronic diseases that each individual may have.
The epidemic model in \cite{LiYangaZhoub:NARWA:2011} considers a latent stage 
and vaccination of newborns and susceptible. In \cite{NareshEtall:AMC:2006}, 
it is assumed that the HIV epidemic spreads both through horizontal 
and vertical transmission; in \cite{NareshEtall:MCM:2009}, the immigration 
of infective individuals is considered, both models with a variable 
size population. The effect of screening unaware infective individuals 
on the spread of HIV, in a constant population, is considered 
in the mathematical model proposed in \cite{TripathiEtall:AMC:2007}. 
In \cite{Cai:etall:JCAM:2009}, the global stability is studied 
for a HIV/AIDS model with two infective stages and where a discrete time
delay is introduced, describing the time from start of treatment 
in the symptomatic stage until treatment effects become visible. 

Motivated by the results of \cite{SilvaTorres:TBHIV:2015},
in this paper we propose a mathematical model for HIV/AIDS transmission 
with varying population size in a homogeneously mixing population. 
Differently from \cite{SilvaTorres:TBHIV:2015}, here we consider 
a mass action hypothesis for the transmission rate. We assume 
that the rate at which susceptible are infected by individuals 
with AIDS symptoms is bigger or equal than the rate of infection 
by contact with HIV-infected individuals (pre-AIDS). This is justifiable
because individuals with AIDS symptoms have a higher viral load 
and it is known that there exists a positive correlation between 
viral load and infectiousness \cite{art:viral:load}. On the other hand, 
individuals with HIV-infection under anti-retroviral treatment (ART) 
suffer a partial restoration of the immune function and, therefore, 
we assume that the rate of infection by contact with individuals 
under ART is smaller or equal than the rate of infection by contact 
with HIV-infected individuals (pre-AIDS), which are not under ART 
(see, e.g., \cite{AIDS:chronic:Lancet:2013}). We prove the global 
stability of the disease free equilibrium whenever the basic 
reproduction number $R_0$ is less than one; and the global stability 
of the unique endemic equilibrium when $R_0$ is greater than one. 
The global stability analysis is done through Lyapunov's direct 
method combined with LaSalle's invariance principle.  

The paper is organized as follows. In Section~\ref{sec:model}, 
we describe the mathematical model for HIV/AIDS transmission. 
Then, in Section~\ref{sec:dfe}, we prove existence and global stability 
of the disease free equilibrium. The existence and global stability 
of the unique endemic equilibrium point is proved in Section~\ref{sec:ee}. 
The stability results are then illustrated through numerical simulations 
in Section~\ref{sec:numsimu}. We finish the paper with Section~\ref{sec:conc} 
of concluding remarks.  


\section{Model for HIV/AIDS transmission}
\label{sec:model}

In this paper, we propose and analyze a mathematical model 
for HIV/AIDS transmission with varying population size 
in a homogeneously mixing population. 
The model is based on that of \cite{SilvaTorres:TBHIV:2015}, 
and subdivides the human population into four mutually-exclusive 
compartments: susceptible individuals ($S$); 
HIV-infected individuals with no clinical symptoms of AIDS 
(the virus is living or developing in the individuals 
but without producing symptoms or only mild ones) 
but able to transmit HIV to other individuals ($I$); 
HIV-infected individuals under ART treatment (the so called 
chronic stage) with a viral load remaining low ($C$); 
and HIV-infected individuals with AIDS clinical symptoms ($A$).
The total population at time $t$, denoted by $N(t)$, is given by
\begin{equation*}
N(t) = S(t) + I(t) + C(t) + A(t).
\end{equation*} 
The effective contact with people infected with HIV is
at a rate $\lambda$, given by
\begin{equation*}
\lambda = \beta \left( I + \eta_C \, C  + \eta_A  A \right),
\end{equation*}
where $\beta$ is the contact rate for HIV transmission.
The modification parameter $\eta_A \geq 1$ accounts for the relative
infectiousness of individuals with AIDS symptoms, in comparison to those
infected with HIV and no AIDS symptoms. Individuals with AIDS symptoms
are more infectious than HIV-infected individuals (pre-AIDS) because
they have a higher viral load and there is a positive correlation
between viral load and infectiousness \cite{art:viral:load}. 
On the other hand, $\eta_C \leq 1$ translates the partial restoration 
of the immune function of individuals with HIV infection 
that are correctly treated under ART \cite{AIDS:chronic:Lancet:2013}.

We assume that HIV-infected individuals 
with and without AIDS symptoms have access to ART treatment. 
HIV-infected individuals with no AIDS symptoms, $I$, progress to the class 
of individuals with HIV infection under ART treatment, $C$, at a rate $\phi$, 
and HIV-infected individuals with AIDS symptoms are treated for HIV at a rate $\alpha$.
We assume that HIV-infected individuals with AIDS symptoms, $A$, 
that start treatment, move to the class of HIV-infected individuals, $I$, 
and will move to the chronic class, $C$, only if the treatment is maintained. 
HIV-infected individuals with no AIDS symptoms, $I$, that do not take 
ART treatment, progress to the AIDS class, $A$, at rate $\rho$. We assume 
that only HIV-infected individuals with AIDS symptoms, $A$, 
suffer from an AIDS induced death, at a rate $d$. These assumptions 
are translated into the following mathematical model:
\begin{equation}
\label{eq:model:1}
\begin{cases}
\dot{S}(t) = \Lambda - \beta \left( I(t) + \eta_C \, C(t)  
+ \eta_A  A(t) \right) S(t) - \mu S(t),\\[0.2 cm]
\dot{I}(t) = \beta \left( I(t) + \eta_C \, C(t)  
+ \eta_A  A(t) \right) S(t) - \left(\rho + \phi + \mu\right) I(t) 
+ \omega C(t) + \alpha A(t), \\[0.2 cm]
\dot{C}(t) = \phi I(t) - (\omega + \mu)C(t),\\[0.2 cm]
\dot{A}(t) =  \rho \, I(t) - (\alpha + \mu + d) A(t).
\end{cases}
\end{equation}
From $N(t) = S(t) + I(t) + C(t) + A(t)$ and \eqref{eq:model:1}, 
it follows that 
$$
\dot{N}(t) = \Lambda - \mu N(t) - d A(t).
$$
Thus, the total population size $N$ may vary in time. 
Let $\Omega$ denote the biologically feasible region
\begin{equation*}
\Omega = \left\{ \left( S, I, C, A \right) \in \mathbb{R}_+^{4} \, 
: \, N \leq \Lambda/\mu \right\}.
\end{equation*}
Using a standard comparison theorem (see \cite{Lakshmikantham:1989}), 
one can easily show that $N(t) \leq \frac{\Lambda}{\mu}$ 
if $N(0) \leq \frac{\Lambda}{\mu}$. Thus, the region $\Omega$ 
is positively invariant. Hence, it is sufficient to consider
the dynamics of the flow generated by \eqref{eq:model:1} in $\Omega$.
In this region, the model is epidemiologically and mathematically 
well posed in the sense of \cite{Hethcote:2000}. In other words,
every solution of the model \eqref{eq:model:1} with initial conditions
in $\Omega$ remains in $\Omega$ for all $t > 0$. Therefore, 
the dynamics of our model will be considered in $\Omega$.


\section{Existence and global stability of the DFE}
\label{sec:dfe}

Model \eqref{eq:model:1} has a disease-free equilibrium (DFE) given by
\begin{equation}
\label{eq:DFE}
\Sigma_0 = \left( S^0, I^0, C^0, A^0  \right) 
= \left(\frac{\Lambda}{\mu},0, 0,0  \right).
\end{equation}
Following \cite{van:den:Driessche:2002}, the basic reproduction number 
$R_{0}$ for \eqref{eq:model:1}, which represents the expected 
average number of new HIV infections produced by a single HIV-infected 
individual when in contact with a completely susceptible population, is given by
\begin{equation*}
R_0 = \frac{ S^0 \beta\, \left(  \xi_2  \left( \xi_1 +\rho\, \eta_A \right) 
+ \eta_C \,\phi \, \xi_1 \right) }{\mu\, \left(  \xi_2  \left( \rho + \xi_1\right) 
+\phi\, \xi_1 +\rho\,d \right) +\rho\, \omega\,d} 
= \frac{S^0 \mathcal{N}}{\mathcal{D}},
\end{equation*}
where $\xi_1 = \alpha + \mu + d$, 
$\xi_2 = \omega + \mu$, 
$\mathcal{N} = \beta\, \left(  \xi_2  \left( \xi_1 +\rho\, \eta_A \right) 
+ \eta_C \,\phi \, \xi_1 \right)$, and 
$$
\mathcal{D} = \mu\, \left(  \xi_2  \left( \rho + \xi_1\right) 
+\phi\, \xi_1 +\rho\,d \right) +\rho\, \omega\,d.
$$
The following local stability result follows easily 
from Theorem~2 of \cite{van:den:Driessche:2002}.

\begin{lemma}
The disease free equilibrium $\Sigma_0$ is locally asymptotically 
stable if $R_0 < 1$ and unstable if $R_0 > 1$.
\end{lemma}

Now we prove the global stability of the disease free equilibrium \eqref{eq:DFE}. 

\begin{theorem}
\label{theo:global:stab:dfe}
The disease free equilibrium $\Sigma_0$ is globally 
asymptotically stable for $R_0 < 1$. 
\end{theorem}

\begin{proof}
Let $\xi_3 = \rho + \phi + \mu$. 	
Consider the following Lyapunov function:
\begin{equation*}
\begin{split}
V &= \left( \xi_1 \xi_2 + \xi_1 \phi \eta_C + \xi_2 \rho \eta_A \right) I 
+ \left( \xi_1 \omega + \xi_1 \xi_3 \eta_C + \rho \eta_A \omega 
- \eta_C \rho \alpha \right) C \\
&\quad + \left( \alpha \xi_2 + \xi_2 \xi_3 \eta_A 
+ \phi \eta_C \alpha - \phi \eta_A \omega \right) A.
\end{split}
\end{equation*}
Note that $\xi_1 \omega + \xi_1 \xi_3 \eta_C + \rho \eta_A \omega 
- \eta_C \rho \alpha = \xi_1 \omega + \alpha (\phi + \mu)\eta_C 
+ (\mu + d) \xi_3 \eta_C + \rho \eta_A \omega > 0$ and 
$\alpha \xi_2 + \xi_2 \xi_3 \eta_A + \phi \eta_C \alpha 
- \phi \eta_A \omega = \alpha \xi_2 + \omega (\rho + \mu)\eta_A 
+ \mu \xi_3 \eta_A + \phi \eta_C \alpha > 0$. 
The time derivative of $V$ computed along the solutions 
of \eqref{eq:model:1} is given by
\begin{equation*}
\begin{split}
\dot{V} &= \left( \xi_1 \xi_2 + \xi_1 \phi \eta_C 
+ \xi_2 \rho \eta_A \right) \dot{I} + \left( \xi_1 \omega 
+ \xi_1 \xi_3 \eta_C + \rho \eta_A \omega 
- \eta_C \rho \alpha \right) \dot{C}\\
&\quad + \left( \alpha \xi_2 + \xi_2 \xi_3 \eta_A 
+ \phi \eta_C \alpha - \phi \eta_A \omega \right) \dot{A}\\
&= \left( \xi_1 \xi_2 + \xi_1 \phi \eta_C + \xi_2 \rho \eta_A \right) 
\left( \beta \left( I + \eta_C \, C  + \eta_A  A \right) S 
- \xi_3 I + \alpha A + \omega C \right)\\
& \quad + \left( \xi_1 \omega + \xi_1 \xi_3 \eta_C 
+ \rho \eta_A \omega - \eta_C \rho \alpha \right) 
\left( \phi I - \xi_2 C \right) \\
&\quad + \left( \alpha \xi_2 + \xi_2 \xi_3 \eta_A 
+ \phi \eta_C \alpha - \phi \eta_A \omega \right) 
\left( \rho \, I - \xi_1 A \right),
\end{split}
\end{equation*}
which can be further simplified to 
\begin{equation*}
\begin{split}
\dot{V} &= (\xi_1 \xi_2 \beta +  \xi_1 \phi \eta_C \beta  
+ \xi_2 \rho \eta_A \beta) I S + (- \xi_1 \xi_2 \xi_3 
+ \xi_1 \omega \phi + \alpha \xi_2 \rho ) I\\
& \quad + \eta_C (\xi_1 \xi_2 \beta + \xi_1 \phi \eta_C \beta 
+ \xi_2 \rho \eta_A \beta ) C S + \eta_C (- \xi_1 \xi_3 \xi_2 
+ \xi_1 \phi \omega +  \rho \alpha \xi_2) C\\
& \quad + \eta_A (\xi_1 \xi_2 \beta + \xi_1 \phi \eta_C \beta 
+ \xi_2 \rho \eta_A \beta ) A S + \eta_A (- \xi_2 \xi_3 \xi_1 
+ \phi \omega \xi_1 + \xi_2 \rho \alpha ) A.
\end{split}
\end{equation*}
As $S \leq S^0$, the following inequality holds:
\begin{equation*}
\begin{split}
\dot{V} &\leq  (\xi_1 \xi_2 \beta +  \xi_1 \phi \eta_C \beta  
+ \xi_2 \rho \eta_A \beta) I S^0 + (- \xi_1 \xi_2 \xi_3 
+ \xi_1 \omega \phi + \alpha \xi_2 \rho ) I\\
& \quad + \eta_C (\xi_1 \xi_2 \beta + \xi_1 \phi \eta_C \beta 
+ \xi_2 \rho \eta_A \beta ) C S^0 + \eta_C 
\left(- \xi_1 \xi_3 \xi_2 + \xi_1 \phi \omega +  \rho \alpha \xi_2\right) C\\
& \quad + \eta_A (\xi_1 \xi_2 \beta + \xi_1 \phi \eta_C \beta 
+ \xi_2 \rho \eta_A \beta ) A S^0 + \eta_A 
\left(- \xi_2 \xi_3 \xi_1 + \phi \omega \xi_1 + \xi_2 \rho \alpha\right) A.
\end{split}
\end{equation*}
From  $S^0 \left(\xi_1 \xi_2 \beta +  \xi_1 \phi \eta_C \beta  
+ \xi_2 \rho \eta_A \beta\right) = \mathcal{N}$ and 
$- \xi_1 \xi_2 \xi_3 + \xi_1 \omega \phi + \alpha \xi_2 \rho = - \mathcal{D}$, 
\begin{equation*}
\begin{split}
\dot{V} &\leq \mathcal{N} I - \mathcal{D} I  + \eta_C (\mathcal{N} C 
- \mathcal{D} C )+ \eta_A (\mathcal{N} A - \mathcal{D} A )\\
&= \mathcal{D} I \left( R_0 - 1 \right)  + \eta_C \mathcal{D} 
C \left( R_0 - 1 \right) + \eta_A \mathcal{D} A \left( R_0 - 1 \right)\\
&\leq 0 \, \, \text{for} \, \, R_0 < 1.
\end{split}
\end{equation*}
Because all model parameters are nonnegative, it follows that 
$\dot{V} \leq 0$, for $R_0 < 1$ with $\dot{V}=0$, if and only if 
$I =C=A=0$. Substituting $(I, C, A) = (0, 0, 0)$ into \eqref{eq:model:1} 
shows that $S \to S^0=\frac{\Lambda}{\mu}$ as $t \to \infty$. Hence, 
$V$ is a Lyapunov function on $\Omega$ and the largest compact invariant 
set in $\{ (S, I, C, A) \in \Omega \, : \, \dot{V} =0 \}$ 
is the singleton $\{ \Sigma_0\}$. Thus, by LaSalle's invariance principle \cite{LaSalle1976}, 
every solution of \eqref{eq:model:1}, with initial conditions in $\Omega$, 
approaches $\Sigma_0$ as $t \to \infty$, whenever $R_0 < 1$. 	
\end{proof}


\section{Existence and global stability of the endemic equilibrium}
\label{sec:ee}

It is easy to show that model \eqref{eq:model:1} 
has a unique endemic equilibrium 
$$
\Sigma_+ =(S^*, I^*, C^*, A^*)
$$ 
whenever $R_0 > 1$. This is precisely 
stated in Lemma~\ref{lem:uni:ee}.  

\begin{lemma}
\label{lem:uni:ee}
The model \eqref{eq:model:1} has a unique endemic equilibrium 
$\Sigma_+ =(S^*, I^*, C^*, A^*)$ whenever $R_0 > 1$, 
which is given by
\begin{equation*}
S^* = \frac{ \mathcal{D}}{ \mathcal{N}}, \quad 
I^* = \frac{\xi_1 \xi_2 (\Lambda \mathcal{N} 
- \mu \mathcal{D})}{\mathcal{D} \mathcal{N}}, \quad
C^* = \frac{\phi \xi_1 (\Lambda \mathcal{N} 
-\mu \mathcal{D})}{\mathcal{D} \mathcal{N}}, \quad
A^* = \frac{\rho \xi_2 (\Lambda \mathcal{N} 
- \mu \mathcal{D})}{\mathcal{D} \mathcal{N}}.
\end{equation*}	
\end{lemma} 

We now prove the global stability of the endemic equilibrium $\Sigma_+$. 

\begin{theorem}
\label{theo:globstab:ee}
The endemic equilibrium $\Sigma_+$ of model \eqref{eq:model:1} 
is globally asymptotically stable for $R_0 > 1$.
\end{theorem}

\begin{proof}
We start by defining the region 
$\Omega_0 = \left\{ (S, I, C, A) \in \Omega \, | \,  I = C = A = 0 \right\}$.
Consider the following Lyapunov function:
\begin{equation*}
V = \left(S - S^* \ln(S)\right) + \left(I - I^* \ln(I)\right) 
+ \frac{\omega}{\xi_2} \left(C - C^* \ln(C)\right) 
+ \frac{\alpha}{\xi_1} \left(A - A^* \ln(A)\right).
\end{equation*}
Differentiating $V$ with respect to time gives
\begin{equation*}
\dot{V} = \left(1-\frac{S^*}{S}\right)\dot{S} + \left(1-\frac{I^*}{I}\right)\dot{I} 
+ \frac{\omega}{\xi_2} \left(1-\frac{C^*}{C}\right)\dot{C} + \frac{\alpha}{\xi_1} 
\left(1-\frac{A^*}{A}\right)\dot{A}.
\end{equation*}
Substituting the expressions for the derivatives in $\dot{V}$, 
it follows from \eqref{eq:model:1} that
\begin{equation}
\label{eq:difV:1}
\begin{split}
\dot{V} &= \left(1-\frac{S^*}{S}\right)\left[ 
\Lambda - \beta \left( I + \eta_C \, C  + \eta_A  A \right) S - \mu S \right]\\
&\quad + \left(1-\frac{I^*}{I}\right)\left[  \beta \left( I + \eta_C \, C  
+ \eta_A  A \right) S - \xi_3 I + \alpha A + \omega C \right]\\
&\quad + \frac{\omega}{\xi_2}  \left(1-\frac{C^*}{C}\right)\left[\phi I-\xi_2 C \right]
+ \frac{\alpha}{\xi_1} \left(1-\frac{A^*}{A}\right)\left[  \rho I - \xi_1 A \right].
\end{split}
\end{equation}
Using the relation $\Lambda = \beta \left( I^* + \eta_C \, C^*  
+ \eta_A  A^* \right) S^* + \mu S^*$, we have from the first 
equation of system \eqref{eq:model:1} at steady-state 
that \eqref{eq:difV:1} can be written as
\begin{equation*}
\begin{split}
\dot{V} &= \left(1-\frac{S^*}{S}\right)\left[ 
\beta \left( I^* + \eta_C \, C^*  + \eta_A  A^* \right) S^* 
+ \mu S^* - \beta \left( I + \eta_C \, C  + \eta_A  A \right) S - \mu S \right]\\
&\quad + \left(1-\frac{I^*}{I}\right)\left[  \beta \left( I + \eta_C \, C  
+ \eta_A  A \right) S - \xi_3 I + \alpha A + \omega C \right]\\
&\quad + \frac{\omega}{\xi_2}  \left(1-\frac{C^*}{C}\right)\left[  \phi I -  \xi_2 C \right] 
+ \frac{\alpha}{\xi_1} \left(1-\frac{A^*}{A}\right)\left[  \rho I - \xi_1 A \right],
\end{split}
\end{equation*}
which can then be simplified to
\begin{equation*}
\begin{split}
\dot{V} &= \left(1-\frac{S^*}{S}\right) \beta I^* S^* 
+ \mu S^* \left( 2 - \frac{S}{S^*} - \frac{S^*}{S}\right) - \beta I S + \beta I S^*\\
&\quad + \beta( \eta_C C^* + \eta_A A^*) S^* - \beta (\eta_C C + \eta_A A) S 
- \frac{S^*}{S} \beta (\eta_C C^* + \eta_A A^*) S^* \\
&\quad + S^* \beta (\eta_C C + \eta_A A)
+ \left(1-\frac{I^*}{I}\right)\left[  \beta \left( I + \eta_C \, C  + \eta_A  A \right) 
S - \xi_3 I + \alpha A + \omega C \right]\\
&\quad + \frac{\omega}{\xi_2}  \left(1-\frac{C^*}{C}\right)\left[\phi I - \xi_2 C \right] 
+ \frac{\alpha}{\xi_1} \left(1-\frac{A^*}{A}\right)\left[  \rho I - \xi_1 A \right].
\end{split}
\end{equation*}
Using the relations at the steady state
\begin{equation*}
\xi_3 I^* = \beta (I^* + \eta_C C^* + \eta_A A^*) S^*  + \alpha A^*  + \omega C^*, 
\quad \xi_2 C^* = \phi I^*, 
\quad \xi_1 A^* = \rho I^*,
\end{equation*} 
and after some simplifications, we have
\begin{equation*}
\begin{split}
\dot{V} &= \left( \beta I^* S^* + \mu S^* \right) 
\left(2 - \frac{S}{S^*} -\frac{S^*}{S} \right) 
+ \beta S^*\left( \eta_C C^* + \eta_A A^* \right) 
\left( 2 - \frac{S^*}{S} - \frac{I}{I^*} \right) \\
&\quad + \beta S^* \left( \eta_C C + \eta_A A \right) 
\left( 1 - \frac{I^*}{I} \frac{S}{S^*}\right)     
+ \alpha A^* \left( 1 - \frac{A}{A^*} \frac{I^*}{I} \right) 
+ \omega C^* \left( 1 - \frac{C}{C^*} \frac{I^*}{I} \right)\\
&\quad +  \frac{\omega \phi}{\xi_2} I^* \left( 1 - \frac{I}{I^*} \frac{C^*}{C} \right)
+ \frac{\alpha \rho}{\xi_1} I^* \left( 1 - \frac{I}{I^*} \frac{A^*}{A} \right).
\end{split}
\end{equation*}
Because the geometric mean is less or equal than the arithmetic mean,
it follows that the terms between the larger brackets are less  
or equal than zero and $\dot{V} = 0$ holds if and only if $(S, I, C, A)$ 
take the equilibrium values $(S^*, I^*, C^*, A^*)$. Thus, 
by LaSalle's invariance principle, the endemic equilibrium 
$\Sigma_+$ is globally asymptotically stable. 
\end{proof}
  

\section{Numerical simulations}
\label{sec:numsimu}

In this section, we provide some numerical simulations 
that illustrate the analytic results proved in Sections~\ref{sec:dfe}
and \ref{sec:ee}. Consider the parameter values $\mu = 1/70$, 
$\Lambda = 2$, $\beta = 0.001$, $\eta_C = 0.04$, $\eta_A = 1.3$, 
$\omega = 0.09$, $\rho = 0.1$, $\phi = 1$, $\alpha = 0.33$ and $d = 1$. 
The corresponding basic reproduction number is equal to $R_0 = 0.9141$. 
The disease free equilibrium is given by 
$\left(S^0, I^0, C^0, A^0 \right) = (140, 0, 0, 0)$. 
Figure~\ref{fig:glob:stab:dfe} illustrates the stability of the disease 
free equilibrium proved in Theorem~\ref{theo:global:stab:dfe}. 
\begin{figure}[!htb]
\centering	
\subfloat[\footnotesize{$(S, I)$}]{\label{SI:globstab:dfe}
\includegraphics[width=0.45\textwidth]{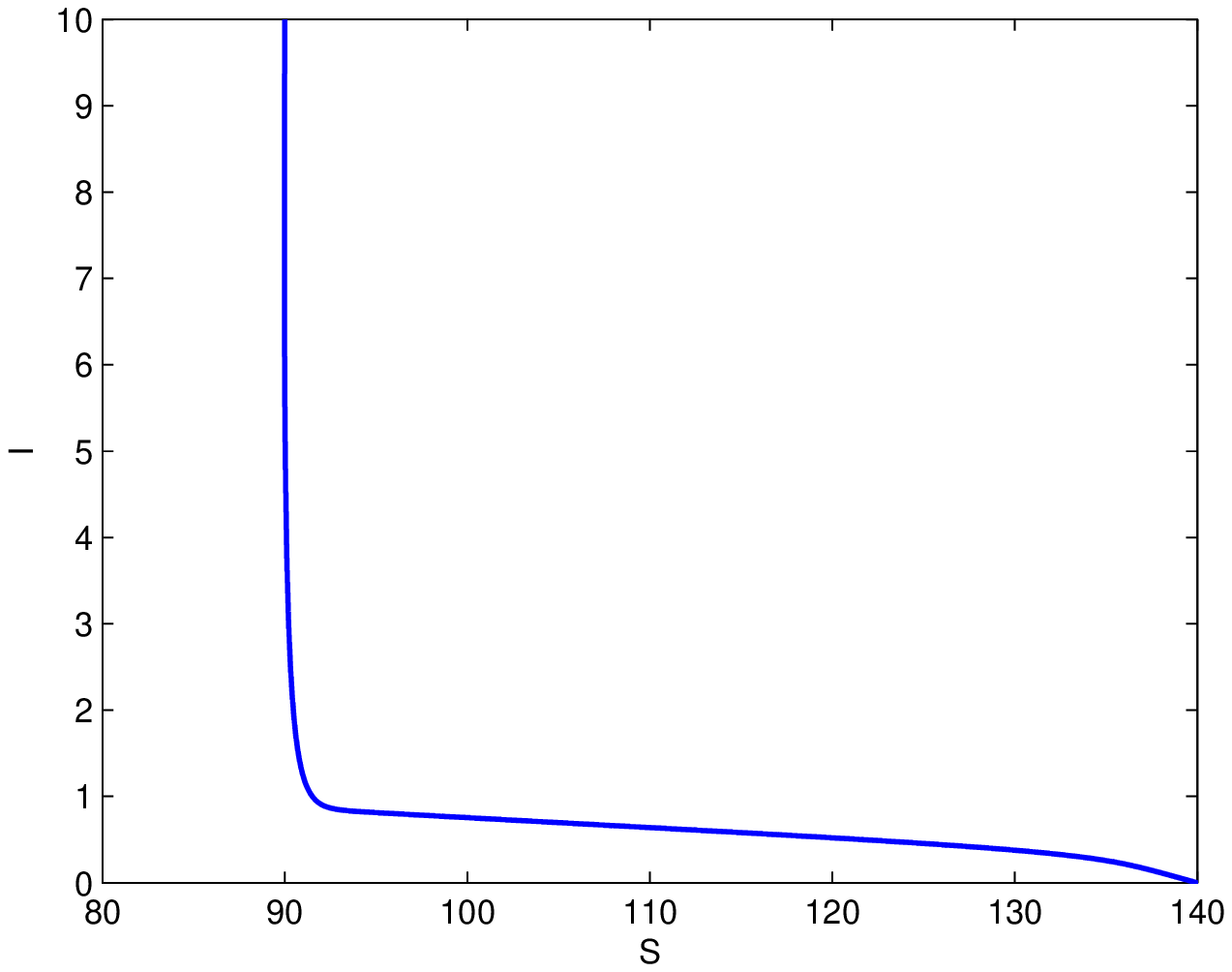}}
\subfloat[\footnotesize{$(C, A)$}]{\label{CA:globstab:dfe}
\includegraphics[width=0.45\textwidth]{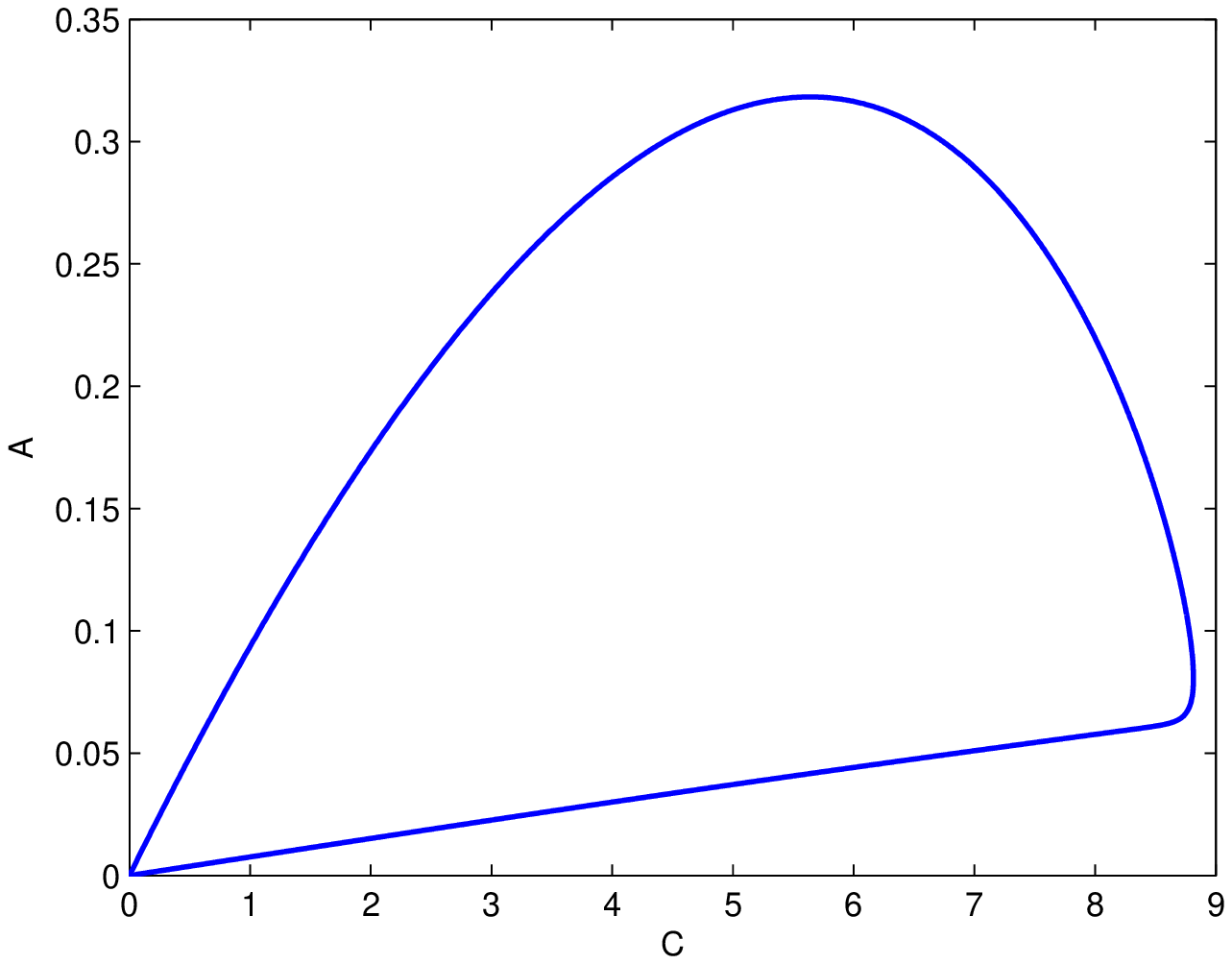}}
\caption{Global stability of the disease free equilibrium \eqref{eq:DFE} 
for $\mu = 1/70$, $\Lambda = 2$, $\beta = 0.001$, $\eta_C = 0.04$, 
$\eta_A = 1.3$, $\omega = 0.09$, $\rho = 0.1$, 
$\phi = 1$, $\alpha = 0.33$ and $d = 1$.}
\label{fig:glob:stab:dfe}
\end{figure}
In Figure~\ref{fig:glob:stab:ee}, we can observe the stability 
of the endemic equilibrium proved in Theorem~\ref{theo:globstab:ee} 
for the paremeter values $\mu = 1/70$, $\Lambda = 2$, $\beta = 0.002$,    
$\eta_C = 0.04$, $\eta_A = 1.3$, $\omega = 0.09$, $\rho = 0.1$, 
$\phi = 1$, $\alpha = 0.33$ and $d = 1$, which corresponds to a basic
reproduction number equal to $R_0 = 1.8281$ and where the unique 
endemic equilibrium is given by $\Sigma_+ =(S^*, I^*, C^*, A^*) 
= (76.5820, 3.9959, 38.3171, 0.2973)$. 
\begin{figure}[!htb]
\centering	
\subfloat[\footnotesize{$(S, I)$}]{\label{S:globstab:ee}
\includegraphics[width=0.45\textwidth]{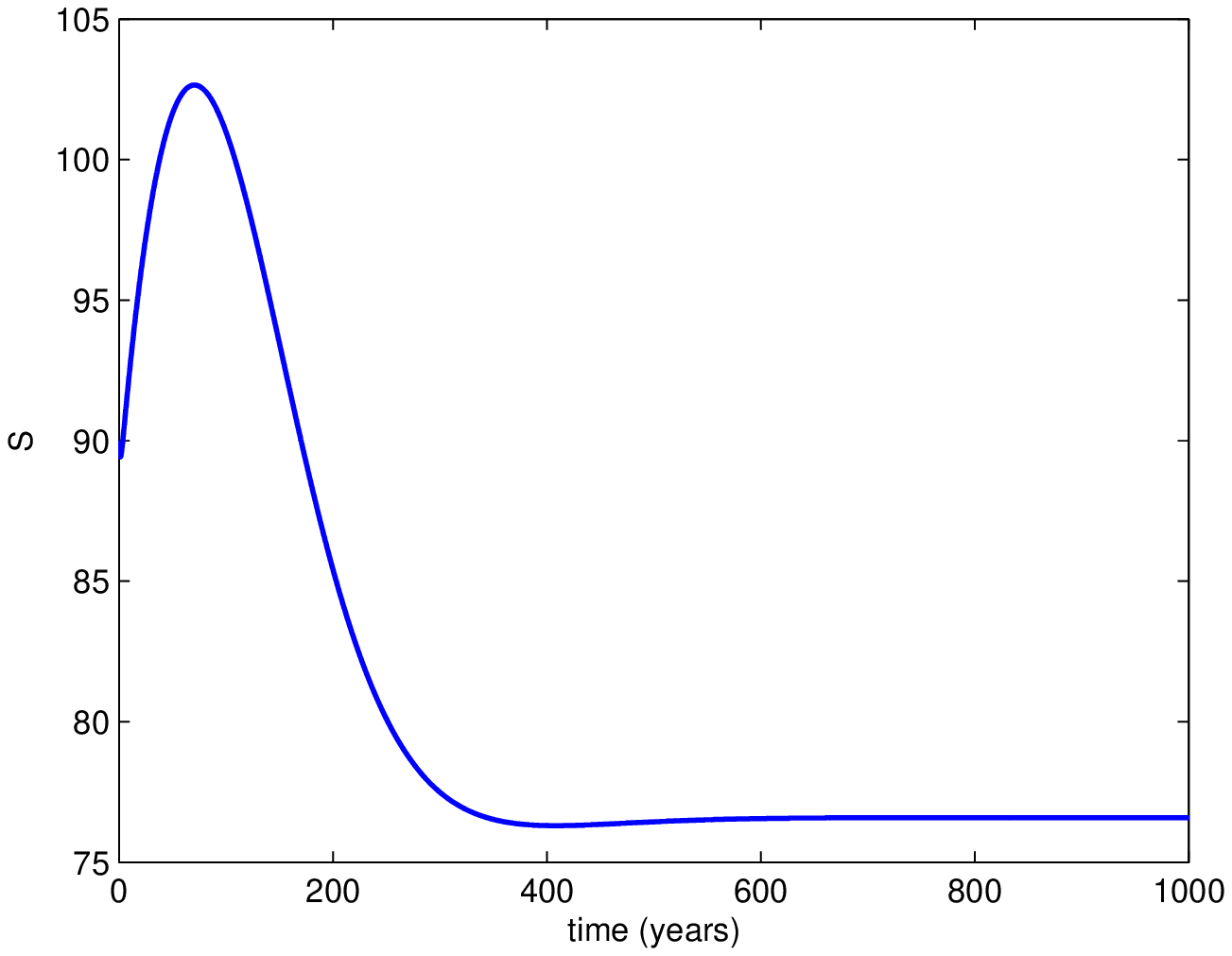}}
\subfloat[\footnotesize{$(C, A)$}]{\label{CA:globstab:ee}
\includegraphics[width=0.45\textwidth]{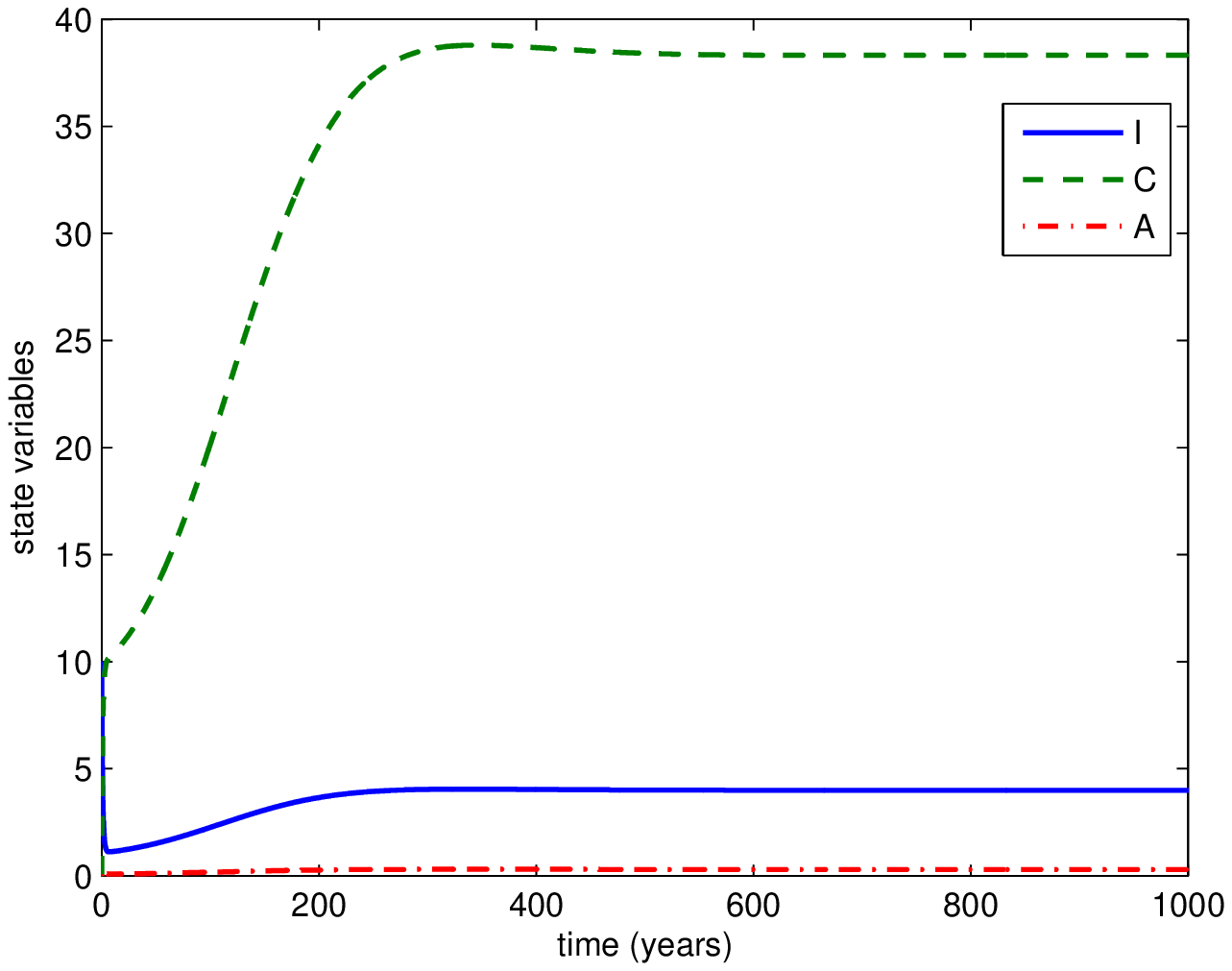}}
\caption{Global stability of the endemic equilibrium of
Lemma~\ref{lem:uni:ee} for $\mu = 1/70$, $\Lambda = 2$, 
$\beta = 0.002$, $\eta_C = 0.04$, $\eta_A = 1.3$, $\omega = 0.09$, 
$\rho = 0.1$, $\phi = 1$, $\alpha = 0.33$ and $d = 1$.}
\label{fig:glob:stab:ee}
\end{figure}


\section{Conclusion}
\label{sec:conc}

We proposed a mathematical model for HIV/AIDS transmission 
with variable total population size and different transmission rates 
depending on the viral load of HIV infected individuals. We proved 
existence of a disease free equilibrium and computed 
the basic reproduction number $R_0$ using the method 
in \cite{van:den:Driessche:2002}. Existence of 
an endemic equilibrium is proved for $R_0 > 1$. We also proved 
the global stability of the disease free equilibrium when $R_0 < 1$ 
and the global stability of the endemic equilibrium for $R_0 > 1$. 
The proofs of global stability are carried out through Lyapunov's 
direct method combined with LaSalle's invariance principle. 
The numerical simulations provided in Section~\ref{sec:numsimu} 
illustrate the obtained stability results. 


\section*{Acknowledgments}

This research was partially supported by FCT
(The Portuguese Foundation for Science and Technology)
within projects UID/MAT/04106/2013, CIDMA, 
and PTDC/EEI-AUT/2933/2014, TOCCATA, 
funded by FEDER funds through COMPETE 2020 -- 
Programa Operacional Competitividade e
Internacionaliza\c{c}\~ao (POCI) 
and by national funds through FCT.
Silva is also grateful to the FCT post-doc 
fellowship SFRH/BPD/72061/2010.



\end{document}